\documentclass[runningheads,letterpaper,oribibl]{llncs}

\usepackage{amssymb}
\usepackage{amsmath}
\usepackage{amsfonts}  
\usepackage{xspace}
\usepackage{graphicx,epsfig}
\usepackage{url}
\usepackage{stmaryrd}
\usepackage{algorithm} 
\usepackage{algpseudocode}

\usepackage[dvipsnames,usenames]{xcolor}
\usepackage[letterpaper=true,colorlinks=true,pdfpagemode=none,urlcolor=RoyalBlue,linkcolor=RoyalBlue,citecolor=OliveGreen,pdfstartview=FitH]{hyperref}
\usepackage{prettyref}

\let\pref=\prettyref
\newcommand{\savehyperref}[2]{\texorpdfstring{\hyperref[#1]{#2}}{#2}}

\newrefformat{eq}{\savehyperref{#1}{\textup{(\ref*{#1})}}}
\newrefformat{chp}{\savehyperref{#1}{Chapter~\ref*{#1}}}
\newrefformat{lem}{\savehyperref{#1}{Lemma~\ref*{#1}}}
\newrefformat{def}{\savehyperref{#1}{Definition~\ref*{#1}}}
\newrefformat{thm}{\savehyperref{#1}{Theorem~\ref*{#1}}}
\newrefformat{cor}{\savehyperref{#1}{Corollary~\ref*{#1}}}
\newrefformat{cha}{\savehyperref{#1}{Chapter~\ref*{#1}}}
\newrefformat{sec}{\savehyperref{#1}{Section~\ref*{#1}}}
\newrefformat{app}{\savehyperref{#1}{Appendix~\ref*{#1}}}
\newrefformat{tab}{\savehyperref{#1}{Table~\ref*{#1}}}
\newrefformat{fig}{\savehyperref{#1}{Figure~\ref*{#1}}}
\newrefformat{hyp}{\savehyperref{#1}{Hypothesis~\ref*{#1}}}
\newrefformat{alg}{\savehyperref{#1}{Algorithm~\ref*{#1}}}
\newrefformat{item}{\savehyperref{#1}{Item~\ref*{#1}}}
\newrefformat{step}{\savehyperref{#1}{step~\ref*{#1}}}
\newrefformat{conj}{\savehyperref{#1}{Conjecture~\ref*{#1}}}
\newrefformat{fact}{\savehyperref{#1}{Fact~\ref*{#1}}}
\newrefformat{prop}{\savehyperref{#1}{Proposition~\ref*{#1}}}
\newrefformat{claim}{\savehyperref{#1}{Claim~\ref*{#1}}}

\newcommand{\RR}{\mathbb R}

\newcommand{\argmax}{\operatorname{argmax}}

\newcommand{\alg}{{\operatorname{ALG}}}
\newcommand{\opt}{{\operatorname{OPT}}}
\newcommand{\argmin}{{\operatorname{argmin}}}
\renewcommand{\min}{{\operatorname{min}}}
\renewcommand{\max}{{\operatorname{max}}}

\newif\ifcomments
\commentstrue
\ifcomments

\else
\newcommand{\remove}[1]{}
\fi

\newif\ifshowqed
\showqedtrue
\begin{document}
\title{Welfare Maximization with Deferred Acceptance Auctions in Reallocation Problems}
\author{
Anthony Kim\thanks{Supported in part by an NSF Graduate Research Fellowship.}}

\institute{
Department of Computer Science, Stanford University, USA \\ 
\email{tonyekim@stanford.edu}}

\maketitle

\begin{abstract}
We design approximate weakly group strategy-proof mechanisms for resource reallocation problems using Milgrom and Segal's deferred acceptance auction framework: the radio spectrum and network bandwidth reallocation problems in the procurement auction setting and the cost minimization problem with set cover constraints in the selling auction setting. Our deferred acceptance auctions are derived from simple greedy algorithms for the underlying optimization problems and guarantee approximately optimal social welfare (cost) of the agents retaining their rights (contracts). In the reallocation problems, we design procurement auctions to purchase agents' broadcast/access rights to free up some of the resources such that the unpurchased rights can still be exercised with respect to the remaining resources. In the cost minimization problem, we design a selling auction to sell early termination rights to agents with existing contracts such that some minimal constraints are still satisfied with remaining contracts. In these problems, while the ``allocated'' agents transact, exchanging rights and payments, the objective and feasibility constraints are on the ``rejected'' agents. 
\end{abstract}

\section{Introduction}\label{sec:intro}

Motivated by the US government's effort to reallocate channels currently allocated for television broadcasting for wireless broadband services, Milgrom and Segal~\cite{MS14} introduced a class of mechanisms called deferred acceptance (DA) auctions for resource reallocation problems. DA auctions greedily choose an allocation by iteratively rejecting the least attractive bid determined by some scoring functions and can be implemented with adaptive reverse greedy algorithms. Milgrom and Segal showed that DA auctions satisfy several important properties: they are strategyproof, weakly group (WG) strategy-proof, can be implemented using ascending clock auctions, and lead to the same outcomes as the complete-information Nash equilibria of corresponding paid-as-bid auctions.

Subsequently, D\"{u}tting et al.~\cite{DGR14} studied the strengths and limitations of DA auctions with respect to achievable approximation guarantees on social welfare in two selling auction design problems. For the knapsack auction problem, they showed a separation between approximation guarantees by DA auctions and by more general strategyproof mechanisms. For the combinatorial auction problem with single-minded bidders, they designed an $O(d)$-approximate DA auction when bidders' desired bundles' sizes are at most $d$ and an $O(\sqrt{m \log m})$-approximate DA auction where $m$ is the number of items. In a different work, D\"{u}tting et al.~\cite{DRT14} studied double auctions for settings in which unit-demand buyers and unit-supply sellers must be matched one-to-one subject to certain constraints. In particular, they showed WG-strategy-proof DA double auctions can be designed by composing two greedy algorithms, one for each side, that use DA rules.

In this paper, we further develop connections between DA auctions and greedy algorithms in the context of the resource reallocation problems that motivated Milgrom and Segal~\cite{MS14}'s DA auction framework. More specifically, we show that several simple greedy approximation algorithms lead to DA auctions with the same approximation guarantees. We consider welfare maximization in the radio spectrum and network bandwidth reallocation problems and the cost minimization problem with set cover constraints.

In the radio spectrum reallocation problem, we (the government) want to reallocate channels currently allocated for television broadcasting for wireless broadband services, effectively reducing the number of channels available for broadcasting, by buying the television stations' broadcasting rights. The reallocation process involves purchasing some of the rights and reassigning the remaining stations with rights into a smaller set of channels; the cleared spectrum will, then, be used for wireless broadband services. The reassignment should be accomplished in a way that respects constraints stipulating that two interfering neighboring stations are not assigned to the same channel. Assuming the stations bid their values of keeping broadcast rights, we want to maximize the social welfare of the stations to keep their rights.

Similarly, in the network bandwidth reallocation problem, we (the network operator) want to reallocate some network connections currently in use for other purposes. We want to buy access rights of some firms using the network such that the demands of firms still with their rights can be served in a smaller network, with the goal of maximizing the social welfare of these latter firms.

In addition, we consider the cost minimization problem with set cover constraints in the selling auction setting. The bidders are looking to terminate their contracts, by paying penalties if necessary, and we (the government) want to agree to such requests while ensuring that some minimal constraints, modeled by the well-known set cover problem, are satisfied. We sell early termination rights to these bidders with the goal of minimizing the social cost, i.e., the total bid value, of those bidders whose requests are not honored.

\paragraph{Our Contributions}
Using Milgrom and Segal~\cite{MS14}'s DA auction framework, we design approximate DA auctions for reallocation problems, in which the ``allocated'' agents transact, exchanging rights and payments, while the objective and feasibility constraints are on the ``rejected'' agents. 
We show simple forward greedy algorithms, not the reverse kind, are sufficient to derive DA auctions that guarantee approximately optimal social welfare (cost) of agents to retain their rights (contracts).
Our DA auctions are computationally efficient and can be computable in polynomial time. More interestingly, the scoring functions of the auctions are algorithmic in nature and might not be expressible in a closed form; they use helper variables to track the progress of allocation.


In the radio spectrum reallocation problem, we design DA auctions that are approximately optimal for certain interference graphs: interval graphs, disk graphs, and bounded degree-$d$ graphs. Unlike the near-optimality result in \cite{MS14} that relies on the existence of a non-trivial partitioning scheme, our approximation ratios depend solely on simple graph parameters. Disk graphs, in particular, are a natural modeling representation of the interference graphs of the stations' circular broadcast ranges. For the network bandwidth reallocation problem, we design an approximate DA auction with scoring functions derived from a result due to Briest et al.~\cite{BKV11}. For the cost minimization problem with set cover constraints, we show that a well-known primal-dual greedy approximation algorithm can be implemented by a DA auction. 

\paragraph{Future Work}
It would be interesting to investigate whether or not we can further improve the results in this paper. For instance, in the radio spectrum reallocation problem, a better performing DA auction with more complex scoring functions might be possible on different kinds of interference graphs. More generally, it would be interesting to have black box results formalizing the conditions under which greedy algorithms can be implemented as DA auctions and comparison-type results comparing DA auctions to other kinds of auctions with different incentive properties in terms of social welfare/cost.

\paragraph{Other Related Work}
Mechanisms derived from greedy algorithms have been studied previously for other various mechanism design problems, but not within the DA auction framework and without consideration of WG strategy-proofness. We refer to these work and the references therein: Lehmann et al.~\cite{LOS02}, Borodin and Lucier~\cite{BL10, BL10-2}, and Briest at al.~\cite{BKV11}.

\section{Preliminaries}\label{sec:prelim}

We consider procurement auctions in the resource reallocation problems; the setting for selling auctions in the cost minimization problem is equivalent but inverted. 

Let $N$ be the set of bidders. Each bidder submits a bid and the auction decides on the allocation and payments such that the bidders either ``win'' or ``lose'', i.e., his bid to supply an item is accepted or not, and only the winning bidders receive payments. We assume that bids $b = (b_1, \ldots, b_{|N|})$ are from the bid profile space $B = B_1 \times \cdots \times B_{|N|}$, and that each bidder $i$'s value $v_i$ is in the range $[0, \bar{v}_i]$ and his bid $b_i$ is restricted to a finite set $B_i$ such that $\max ~ B_i > \bar{v}_i$.\footnote{The restriction of finite bid spaces can be removed.} 
A procurement auction has allocation rule $a: B \rightarrow 2^{N}$ and payment rule $p: B \rightarrow \RR^{|N|}$ such that $p_i(b) = 0$ for $i \in N \setminus a(b)$. In settings we consider, there will be constraints on the possible allocations allowed. Note each bidder is a strategic agent that seeks to maximize his utility (or payoff) which is $p_i(b) - v_i$ if he wins, and $0$ otherwise. 

We define important properties of auctions: An auction is {\em strategy-proof} if for every bidder $i$, $v_i \in [0, \bar{v}_i]$, and other bids $b_{-i} \in B_{-i}$, it is optimal for the bidder to truthfully bid his value $v^+_i := \min \{ b_i \in B_i: b_i > v_i \}$.
An auction is {\em weakly group (WG) strategy-proof} if for every profile of values $v$, a set of bidders $S \subseteq N$, and coordinated false bids by these deviating bidders, there exists at least one bidder in $S$ who does not get a strictly better payoff than under truthful reporting.
An auction is {\em $\alpha$-approximate} if it achieves at least the $\alpha$ fraction of the optimal welfare in all problem instances ($\alpha < 1$).

In reallocation problems, we want to design a procurement auction to purchase agents' broadcast/access rights to free up some of the resources such that the unpurchased rights can still be exercised with respect to the remaining resources. Agents report their private values of rights and receive payments when their rights are purchased. Our objective is to maximize the social welfare of agents retaining their rights. Note that the ``allocated'' agents transact, exchanging rights and payments, but the objective and feasibility constraints are on the ``rejected'' agents in these auctions.\footnote{\label{foot:common} In more common auction settings, the objective and feasibility constraints are on the ``allocated'' agents.}

\paragraph{Deferred Acceptance Auctions}

Deferred acceptance (DA) auctions are auctions in which the allocation is determined by an iterative process of rejecting bidders one by one from the whole set. Milgrom and Segal~\cite{MS14} showed they satisfy strong incentive guarantees such as strategy-proofness, WG-strategy-proofness, and other useful implementation properties. Note the well-studied Vickrey auction is known to be not WG-strategy-proof.
\begin{algorithm}[t]
\caption{DA Auction with Scoring Functions $\{s^A_i \}_{A \subseteq N, i \in A}$ }\label{alg:da}
\begin{algorithmic}[1]
\State Accept bids $b_1, \ldots, b_{|N|}$
\State $A = N$
\While{$\exists i \in A: s^A_i(b_i, b_{N \setminus A}) > 0$}
	\State $i = \argmax_{i \in A} s^A_i(b_i, b_{N \setminus A})$
	\State $A = A \setminus \{ i \}$
\EndWhile
\State \Return $A$
\end{algorithmic}
\end{algorithm}

DA auctions can be described as reverse greedy algorithms that use certain scoring functions $\{s^A_i \}_{A \subseteq N, i \in A}$ as shown in Algorithm~\ref{alg:da}. Bidders are removed from $A$ one at a time until termination of the while-loop and the resulting $A$ is the allocation. Note the scoring functions $s^A_i: B_i \times B_{N \setminus A} \rightarrow \RR^+$ are nondecreasing in the first argument. The bidder $i$'s score during an iteration with the current set of active bidders $A$ is dependent on its own bid $b_i$ and bids of those inactive bidders in $N \setminus A$, but not on the bids of other active bidders. For the allocation rule $a$ determined by Algorithm~\ref{alg:da}, the corresponding payment rule $p$ is defined:
\begin{equation}\label{eq:da}
p_i(b_{-i}) = \max \{ b'_i \in B_i: i \in a(b'_i, b_{-i}) \} \enspace, 
\end{equation} 
i.e., the winning bidder $i$'s payment is the maximum bid value with which he remains allocated. For more details, see \cite{MS14}.

We study greedy algorithms with allocation rules implementable by DA auctions. These algorithms include the standard greedy-by-weight algorithms that process elements in order of decreasing weight and several primal-dual greedy algorithms that utilize variables from the linear programming formulations of the underlying problems. In this paper, we are primarily interested in the greedy algorithms of the ``single pass'' nature, i.e., those that start with an empty (and feasible) solution and iteratively augment it without any post-processing steps that might undo some part of the solution.

We formalize a notion to capture which greedy algorithms can be implemented as DA auctions. Note we have an analogous definition for reverse greedy algorithms that instead return $A$ as the final solution\footnote{This is for the more common auction setting described in Footnote~\ref{foot:common}.}, but the following version for forward greedy algorithms is sufficient for the reallocation problems:

\begin{definition} \label{def:daimp}
Let $N$ be the element set and $w: N \rightarrow \RR^+$ be a weight function. A greedy algorithm $\alg$ is {\em DA-implementable} if it can be implemented with active set $A$ which is initialized to $N$ and scoring functions $s^A_i: B_i \times B_{N \setminus A} \rightarrow \RR^+$ for each $A \subseteq N, i \in A$ such that:
\begin{enumerate}
	\item In each iteration, element $i$ to be selected and removed from $A$ is the highest scoring element $\argmax_{i \in A} s^A_i(w_i, w_{N \setminus A})$;
	\item When $s^A_i(w_i, w_{N \setminus A}) = 0$ for all $i \in A$, $\alg$ terminates and returns $N \setminus A$.
\end{enumerate}
\end{definition}

\paragraph{Notations}

We use the common notation $-i$ to denote the bidders other than bidder $i$. We use $x(i)$ and $x_i$ interchangeably to indicate the $i$-th component value (representing the $i$-th bidder, the $i$-th edge, etc.) for any vector variable $x$. Without loss of generality, we assume $\argmax$ and $\argmin$ operators return the lowest indexed argument according to a consistent global order in the case of a tie.


\section{Radio Spectrum Reallocation}\label{sec:spectrum}

We design a simple DA auction that achieves near-optimal social welfare for certain classes of interference graphs in the spectrum reallocation problem. Without relying on the assumption of a partitioning scheme in Milgrom and Segal~\cite{MS14}, we show that our DA auction achieves an approximation ratio dependent only on simple structural parameters of the interference graphs.\footnote{Milgrom and Segal's result assumes the existence of a ordered partition of $N$ into $m$ disjoint sets $N_1, \ldots, N_m$ such that: (1) the edge $(i, j)$ exists for each $i, j \in N_k, 1 \leq k \leq m$; and (2) there exists some $d < n$ such that $|S| + | \cap_{i \in S} \cup_{l < k} \{j \in N_l: (i,j) \in E \} | \leq n$ for each $1 \leq k \leq m$ and $S \subseteq N_k$ with $|S| \leq n - d$.} For our analysis, we proceed with a series of reductions from the spectrum reallocation problem to the maximum weight $k$-colorable subgraph problem to the submodular function maximization and maximum weight independent set problems. 

The spectrum reallocation problem is as follows:
Let $G = (V, E)$ be the interference graph where $V$ is the vertex set representing the television stations and $E$ is the edge set where there is an edge between two vertices if the corresponding stations interfere were they to be assigned to the same channel; we use $N$ and $V$ interchangeably. Let $k$ be the number of channels available for reassignment. Given stations' bids $b_1, b_2, \ldots$, we want to find a set $A$ of stations to allocate, or buy their broadcast rights, to maximize the welfare of those retaining their rights such that the subgraph $G(V \setminus A, E)$ induced by them is $k$-colorable, i.e., colorable with $k$ colors. We interpret the bids to be the stations' reported values of retaining the broadcast rights.

We consider the following classes of $G$: interval graphs, disk graphs, and bounded degree-$d$ graphs. The interval graphs and disk graphs are intersection graphs with a geometric representation; in interval (disk) graphs, the vertices represent line-segments (disks) in a 1D-space (2D-space) and an edge exists between two vertices if their corresponding representations intersect or even just touch. In bounded degree-$d$ graphs, the degree of every vertex is at most $d$.

\subsection{Main Results}

We reduce the spectrum reallocation problem to the NP-hard problem of finding the maximum weight $k$-colorable subgraph of $G$: Given $G(V,E)$, $k$ and a weight function $w: V \rightarrow \RR^+$, we want to find the maximum weight subgraph $V' \subset V$ such that $G(V', E)$ is $k$-colorable. As the weight function is derived from the bids $b_1, b_2, \ldots$, we have the same objective in both problems. In particular, if there is a DA-implementable approximate greedy algorithm for the maximum weight $k$-colorable subgraph problem, we get a DA auction with the same performance guarantee on social welfare for the spectrum reallocation problem.

\begin{algorithm}[t]
\caption{Greedy Algorithm for the Max Weight $k$-colorable Subgraph Prob.}\label{alg:maxkcolor}
\begin{algorithmic}[1]
\State $I_1 = \cdots = I_k = \emptyset$
\For{$v \in V$ in decreasing order of weight}
	\If{$\exists ~ i: I_i \cup \{ v \} \textrm{ remains independent}$}
		\State $i = \min \left\{i : I_i \cup \{ v \} \textrm{ remains independent} \right\}$
		\State $I_i = I_i \cup \{ v \}$
	\EndIf
\EndFor
\State \Return $\bigcup_i I_i$
\end{algorithmic}
\end{algorithm}

In fact, Algorithm~\ref{alg:maxkcolor} is a DA-implementable greedy algorithm with good approximation guarantees:

\begin{theorem}\label{thm:maxkcolor}
Algorithm~\ref{alg:maxkcolor} is a DA-implementable $(1 - e^{- 1 / \alpha})$-approximation algorithm for the maximum weight $k$-colorable subgraph problem where $\alpha$ is $2 + \gamma$ for interval graphs, $(2 + \gamma)^2$ for disk graphs, and $d$ for bounded degree-$d$ graphs, for $\gamma = l_\max / l_\min$, i.e., the ratio between the maximum and minimum lengths (radii) for interval (disk) graphs.
\end{theorem}

\begin{proof}
(DA-implementability)
Note that Algorithm~\ref{alg:maxkcolor} upon termination returns a subgraph along with a valid $k$-coloring given by the independent sets $I_1, \ldots, I_k$. We show its implementation with the active set and scoring functions as follows: 
Let the active set $A$ be the vertices not yet selected and $I^{N \setminus A}_1, \ldots, I^{N \setminus A}_k$ be the associated independent sets consisting of vertices in $N \setminus A$. 
We define the scoring functions:
\begin{equation} \label{eq:maxkcolor}
s^A_v (w_v, w_{N \setminus A}) = \begin{cases}w(v), & \textrm{if $\exists i: I^{N \setminus A}_i \cup \{ v \}$ is independent} \\ 0, & \textrm{otherwise} \end{cases} \enspace.
\end{equation}

Note that each scoring function $s^A_v$ is nondecreasing in the first argument, and that the next element $v$ to be greedily selected into $N \setminus A$ (and out of $A$) is the highest scoring element $\argmax_v s^A_v$.

(Approximation) We defer this part of the proof to Section~\ref{subsec:maxkcolor}.
\end{proof}

Hence, our main result follows. In the special cases of the unit-interval and unit-disk graphs (i.e., $\gamma = 1$), we get constant approximations: 

\begin{corollary}\label{cor:spectrum}
The DA auction with scoring functions~\eqref{eq:maxkcolor} is a $(1 - e^{- 1 / \alpha})$-approximate WG-strategy-proof mechanism for the spectrum reallocation problem, where $\alpha$ is $2 + \gamma$ for interval graphs, $(2 + \gamma)^2$ for disk graphs, and $d$ for bounded degree-$d$ graphs; and $\gamma = l_\max / l_\min$.
\end{corollary}

\subsection{Proof of \pref{thm:maxkcolor} (Approximation)} \label{subsec:maxkcolor}

We show that \pref{alg:maxkcolor} achieves the stated approximation ratios for the interval graphs, disk graphs, and bounded degree-$d$ graphs by analyzing a related algorithm, \pref{alg:maxkcolor2}, which will be shown to be equivalent for a choice of $\alg$. We further reduce the maximum weight $k$-colorable subgraph problem to a monotone submodular function maximization problem. For a short review of submodular functions, we refer to \pref{app:submod}.

Let $M$ be the set of all independent sets of $G$ and $f: 2^M \rightarrow \RR$ be a set function defined as 
$f(S) = \sum_{v \in \cup_{I \in S} I} w(v)$. Note that $f(\emptyset) = 0$ and $f$ is a monotone submodular function. Since we want to find $k$ independent sets $I_1, \ldots, I_k$ (for $k$ colors) such that the total weight of vertices covered by them is maximized, the maximum weight $k$-colorable subgraph problem is equivalent to the maximization problem of 
\begin{equation}\label{eq:submod}
\max_{S \subseteq M : |S| \leq k} f(S) \enspace. 
\end{equation}

\begin{algorithm}[t]  
\caption{A Greedy Algorithm with Subroutine $\alg$} \label{alg:maxkcolor2}
\begin{algorithmic}[1]
\State $V' = V$
\For{$i = 1, \ldots, k$}
	\State Run $\alg$ on $G(V', E)$ and get an (approx.) max weight independent set $I_i$ \label{ln:indep}
	\State $V' = V' \setminus I_i$
\EndFor
\State \Return $\bigcup_i I_i$
\end{algorithmic}
\end{algorithm}

Note the $k$ independent sets should be disjoint to be a valid coloring, but as any $k$ independent sets can be modified to be disjoint, the above maximization problem still gives the same optimal value.

The well-known greedy algorithm due to Nemhauser et al.~\cite{NWF78}, that is \pref{alg:maxkcolor2} with the optimal $\alg$ that returns the maximum weight independent set, is a $(1 - e^{-1})$-approximation algorithm. However, it is not computationally efficient as the maximum weight independent set problem is difficult for general graphs; even the unweighted version is known to be NP-hard and cannot be approximated in polynomial time within a factor of $|V|^{1 - \epsilon}$ for any fixed $\epsilon > 0$, unless P = NP \cite{H99}. 

Instead, we use the following lemma (with its proof in \pref{app:submod}) to show that \pref{alg:maxkcolor2} with an approximation algorithm $\alg$ has a similar approximation guarantee, and design a computationally efficient $\alg$:

\begin{lemma}\label{lem:submod}
\pref{alg:maxkcolor2} with an $\alpha$-approximation algorithm $\alg$ is a $(1 - e^{- 1 / \alpha})$-approximation algorithm for the maximization problem~\eqref{eq:submod}.
\end{lemma}

For the classes of interference graphs in consideration, the following $\alg$ is a polynomial time approximation algorithm:
\begin{multline*} 
\alg := \textrm{Given graph $G$, select vertices in decreasing order of weight} \\ \textrm{as long as those selected form an independent set.} 
\end{multline*}

\begin{lemma}\label{lem:greedy}
Algorithm $\alg$ is an $\alpha$-approximation algorithm for the maximum weight independent set problem, where $\alpha$ is $2 + \gamma$ for interval graphs, $(2 + \gamma)^2$ for disk graphs, and $d$ for bounded degree-$d$ graphs; and $\gamma = l_\max / l_\min$.
\end{lemma}
\begin{proof}
(Sketch) Note the graphs are bounded claw-free graphs; a graph is {\em $(\tau + 1)$-claw free} if each vertex has at most $\tau$ mutually non-adjacent vertices. More generally, $\alg$ is a $\tau$-approximation algorithm on $(\tau + 1)$-claw free graphs. For each vertex selected by the optimal algorithm, $\alg$ either selects it or does not select it in favor of another. For each vertex selected by $\alg$, it can be involved in at most $\tau$ such instances of the latter case. For details, see Appendix~\ref{app:greedy}. 
\end{proof}

Finally, note that \pref{alg:maxkcolor2} is equivalent to \pref{alg:maxkcolor}; we construct the independent sets one by one in the former and all at once in parallel in the latter. By Lemmas~\ref{lem:submod} and \ref{lem:greedy}, \pref{thm:maxkcolor} follows.

\section{Network Bandwidth Reallocation}\label{sec:network}

In this section, we show an approximate DA auction for the network bandwidth reallocation problem. We reduce the problem to the optimization problem of network unicast/multicast routing and show a primal-dual greedy algorithm due to Briest et al.~\cite{BKV11} is DA-implementable. 

The network bandwidth reallocation problem for unicast routing is defined as follows: 
Let $G = (V , E)$ be the network graph with $|V| = n$, $|E| = m$, and edge capacities $c(e), \forall e \in E$. Let $N$ be the firms (the bidders) with access rights such that each firm $i$ has a terminal pair $(s_i, t_i)$ and demand $d(i)$ and private value $v(i)$ for his right. Without loss of generality, we assume that $d(i) \in [0, 1], \forall i$; $c(e) \geq 1, \forall e$; and $C := \min_e c(e) > 1$. Given the reports of firms' values for their rights $b_1, b_2, \ldots$, we buy access rights of some firms such that the demands of those still retaining rights can be satisfied in the network, i.e., there is a feasible solution that routes each unsplittable flow of the demanded amount between terminals subject to the edge capacity constraints. We want to maximize the social welfare of those still holding rights. Note $G$ is the reduced smaller network after removing reallocated edges.

In the multicast routing version of the problem, each firm has a set of terminal vertices, with one being the source, and demands a (unsplittable) multicast tree (a.k.a., Steiner tree) spanning the terminals. 

For the corresponding optimization problem of network routing, we know the values $v(i), \forall i \in N$ and want to compute a subset of firms with the maximum total value such that their demands can be satisfied. Note that the objectives of both mechanism design and algorithmic problems are on the same set of firms. 

\pref{alg:network}, due to Briest et al.~\cite{BKV11}, is a polynomial time greedy algorithm based on the primal-dual scheme (see \pref{app:network} for the primal-dual linear programming relaxations). Let $\mathcal{S}_i$ be the set of all paths from $s_i$ to $t_i$ in $G$ and $\mathcal{S} = \bigcup_i \mathcal{S}_i$; in the multicast routing case, $\mathcal{S}_i$ is the set of all Steiner trees spanning the firm $i$'s terminals. Given $S \in \mathcal{S}_i$, let $q_S(e) = d(i)$ if $e \in S$, and $q_S(e) = 0$ otherwise. Note $\bar{e} \approx 2.718$ is the Euler number. In Line 3, we need to compute the shortest path with respect to the dual variables $y$ in the unicast routing case and the minimum weight Steiner tree in the multicast routing case. We can compute the shortest path exactly using any shortest path algorithm. For the NP-hard Steiner tree problem, we use the polynomial time 1.55-approximation algorithm due to Robins and Zelikovsky~\cite{RZ00}.

\begin{algorithm}[t]
\caption{Greedy Algorithm for the Network Routing Problem} \label{alg:network}
\begin{algorithmic}[1]
\State $\mathcal{T} = \emptyset$; $N' = N$; $y(e) = 1 / c(e), \forall e \in E$
\While{$N' \not = \emptyset$ and $\sum_{e \in E} c(e) y(e) < \bar{e}^{C - 1} m$}
	\State $S_i = \argmin_{S \in \mathcal{S}_i} \sum_{e \in S} y(e), \forall i \in N'$
	\State $i = \argmax_{i \in N'} \left\{ \frac{ v(i) }{ d(i) \cdot \sum_{e \in S_i} y(e) }  \right\}$
	\State $\mathcal{T} = \mathcal{T} \cup \{ S_i \}$; $N' = N' \setminus \{ i \}$
	\State Update $y(e) = y(e) \cdot \left( \bar{e}^{C - 1} m \right)^{q_{S_i}(e) / (c(e) - 1)}, \forall e \in S_i$
\EndWhile 
\State \Return $\mathcal{T}$
\end{algorithmic}
\end{algorithm}

We show that \pref{alg:network} is DA-implementable and obtain a DA auction for the network bandwidth reallocation problem:

\begin{theorem} \label{thm:network}
\pref{alg:network} is a DA-implementable $\left( \frac{\bar{e} \gamma C}{C - 1} m^{1 / (C - 1)} \right)^{-1}$-approximation algorithm for the network bandwidth reallocation problem where $\gamma = 1$ for unicast routing and $\gamma = 1.55$ for multicast routing.
\end{theorem}

\begin{proof}
(DA-implementability)
We show an implementation with an active set and scoring functions as follows:
Let the active set $A$ be the set of firms not yet selected, so $N'$ in \pref{alg:network}, and $\{ y^{N \setminus A}_e \}_{e \in E}$ be the associated dual variables. We define the scoring functions:
\begin{equation}\label{eq:network}
s^A_i(v_i, v_{N \setminus A}) = \frac{v(i)}{d(i) \sum_{e \in S_i} y^{N \setminus A}_e}, \textrm{ where $S_i = \argmin_{S \in \mathcal{S}_i} \sum_{e \in S} y^{N \setminus A}_e$}.
\end{equation}

Note $A$ changes when a firm is selected in Lines 4-5 and the scoring functions change correspondingly when the dual variables $y_e$ change. The scoring functions are nondecreasing in the first argument and functions of attributes of firms in $N \setminus A$. Also, the next firm to be added to $N \setminus A$ is the highest scoring firm. 

(Approximation) We refer this part of the proof to Briest et al.~\cite{BKV11}.
\end{proof}

\begin{corollary}\label{cor:network}
The DA auction with scoring functions~\eqref{eq:network} is a $\left( \frac{\bar{e} \gamma C}{C - 1} m^{1 / (C - 1)} \right)^{-1}$-approximate WG-strategy-proof mechanism for the network bandwidth reallocation problem, where $\gamma = 1$ for unicast routing and $\gamma = 1.55$ for multicast routing.
\end{corollary}

\section{Cost Minimization with Set Cover Constraints}\label{sec:setcover}

We apply our approach to the cost minimization problem with set cover constraints in the selling auction setting.

In the selling auction setting, the DA auction (\pref{alg:da}) is ``inverted'': the while-loop's stopping condition becomes $\exists i \in A: s^A_i(b_i, b_{N \setminus A}) < \infty$ (so, it terminates when all the scores are $\infty$); the next agent to be rejected in Lines 4-5 becomes the lowest scoring agent $\argmin_{i \in A} s^A_i(b_i, b_{N \setminus A})$; and the payment rule~\eqref{eq:da} becomes $p_i(b_{-i}) = \min \{ b'_i \in B_i: i \in a(b'_i, b_{-i}) \}$. Similarly, each bidder wants to maximize his utility which is $v_i - p_i(b)$ if he wins, and $0$ otherwise. 

In the minimization problems, an auction is $\alpha$-approximate if it achieves cost at most $1 / \alpha$ times the optimal cost in all problem instances.

Let $N$ be the set of bidders with cost function $c: N \rightarrow \RR^+$ and $E$ be a universe of elements. For each $i \in N$, there is an associated set $S_i \subseteq E$; we assume each $e \in E$ is covered by, i.e., contained in, at least one $S_i$. For example, $N$ is a set of firms with private costs $c$ that are bidding $b_1, b_2, \ldots$ to prematurely terminate their contracts and $E$ is a representation of their responsibilities. We want to honor requests while rejecting some to ensure that all the responsibilities are covered by at least one rejected firm. We interpret $c(i)$ to be the value of early termination for firm $i$ and want to minimize the social cost, the total bid value, of those still with contracts.

The cost minimization problem with set cover constraints reduces to the well-known set cover problem with the same objective: given the cost function $c$, we want to select a subset of $N$ with minimum cost whose sets cover $E$. The set cover problem includes many other algorithm design problems such as the minimum spanning tree problem, the Steiner tree problem, etc.

\begin{algorithm}[t]
\caption{Greedy Algorithm for the Set Cover Problem} \label{alg:setcover}
\begin{algorithmic}[1]
\State $I = \emptyset$; $y(e) = 0, \forall e \in E$
\While{$\exists e \not \in \bigcup_{i \in I} S_i$}
	\State Increase $y(e)$ until there is some $i$ such that $\sum_{e' \in S_i} y(e') = c(i)$
	\State $I = I \cup \{ i \}$
\EndWhile
\State \Return $I$
\end{algorithmic}
\end{algorithm}

The set cover problem has primal-dual linear programming relaxations (see \pref{app:setcover}) and a polynomial time greedy approximation algorithm, \pref{alg:setcover}. We show that \pref{alg:setcover} is DA-implementable and obtain a DA auction:

\begin{theorem}
\pref{alg:setcover} is a DA-implementable $f^{-1}$-approximation algorithm for the set cover problem where $f = \max_e | \{ i: e \in S_i \} |$.
\end{theorem}

\begin{proof}
(DA-implementability)
We show an implementation with an active set and scoring functions as follows:
Let $A$ be the set of bidders not yet selected and $y^{N \setminus A}(e), \forall e \in E$ be the associated dual variables, dependent on $N \setminus A$, in \pref{alg:setcover}.
We define the scoring functions as follows:
\begin{equation} \label{eq:setcover}
s^A_i(c_i, c_{N \setminus A}) = \begin{cases} c(i) - \sum_{e \in S_i} y^{N \setminus A}(e), & \textrm{if $T^{N \setminus A}_i$ is not empty} \\ \infty, &\textrm{otherwise} \end{cases} \enspace, 
\end{equation}
where $T^{N \setminus A}_i = \{ e : e \in S_i, e \not \in \bigcup_{j \in N \setminus A} S_j \}$, i.e., those elements in $S_i$ not covered by the selected sets in $N \setminus A$. When a bidder $i$ is selected from $A$, there is an element $e \in T^{N \setminus A}_i$, the lowest indexed one if many exist, that we can associate to the bidder and increase the corresponding dual variable $y(e)$ by the amount $c(i) - \sum_{e' \in S_i} y^{N \setminus A}(e')$.

Note that the scoring functions are nondecreasing in the first argument. Assume the lowest scoring bidder $i_t$ is associated with the element $e_t \in T^{N \setminus A}_i$ at each iteration $t$. Then, the steps of the DA auction with the above scoring functions can be realized as the steps of \pref{alg:setcover} when the elements to be used in Line 3 are exactly the associated  elements $e_1, e_2, \ldots$, and the next bidder to be selected by \pref{alg:setcover} is the lowest scoring bidder at each iteration. 

(Approximation) The proof that \pref{alg:setcover} has the approximation ratio of $f$ can be found, for instance, in \cite{WS11}.
\end{proof}

\begin{corollary}
The (``inverted'') DA auction with scoring functions~\eqref{eq:setcover} is a $f^{-1}$-approximate WG-strategy-proof mechanism for the cost minimization problem with set cover constraints where $f = \max_e | \{ i: e \in S_i \} |$.
\end{corollary}

\paragraph{\bf Acknowledgments.}
We would like to thank Vasilis Gkatzelis, Afshin Nikzad, and Amin Saberi for their helpful discussions.

\bibliographystyle{abbrv}
\bibliography{refs}

\appendix

\section{Submodular Function Maximization}\label{app:submod}
We review the basics of submodular functions and give the proof of Lemma~\ref{lem:submod} in Section~\ref{subsec:submod}. 

Let $M$ be a finite ground set and $f: 2^M \rightarrow \RR$ be a set function with $f(\emptyset) = 0$. $f$ is {\em submodular} if  
\[ f(A) + f(B) \geq f(A \cup B) + f(A \cap B), \; \forall A, B \subset M, \]
or equivalently,
\[ f(A \cup \{ i \}) - f(A) \geq f(B \cup \{ i \}) - f(B), \; \forall A \subset B, i \in M. \]
We say that $f$ is {\em monotone} if $f(A) \leq f(B)$ whenever $A \subseteq B$.
To simplify notations, we define $f_S(i)$ to be the {\em marginal value} of element $i$ to set $S$: $f_S(i) = f(S \cup \{ i \}) - f(S)$. Analogously, let $f_S(T)$ for a set $T$ to be $f_S(T) = f(S \cup T) - f(S)$.

Consider the following general maximization problem:
\begin{equation} \label{prob:submod}
 \max_{S \subseteq M : |S| \leq k} f(S) \enspace, 
\end{equation}
where $f$ is a monotone submodular function.

It is well-known that \pref{alg:submodular} due to Nemhauser et al.~\cite{NWF78} is a good approximation algorithm to the above maximization problem: 

\begin{algorithm}[t]
\caption{Greedy Algorithm for the Maximization Prob. \eqref{prob:submod}}\label{alg:submodular} 
\begin{algorithmic}[1]
\State $S = \emptyset$
\For{$j = 1, \ldots, k$}
	\State Let $i = \argmax_{i \in M \setminus S} \{ f(S \cup \{ i \}) - f(S) \}$ \label{alg:submodular:ln:oracle}
	\State $S = S \cup \{ i \}$ 
\EndFor
\State \Return $S$
\end{algorithmic}
\end{algorithm}

\begin{theorem}[Nemhauser et al.~\cite{NWF78}] \label{thm:submod}
Algorithm~\ref{alg:submodular} is a $(1 - e^{-1})$-approximation algorithm for the maximization problem $\max_{|S| \leq k} f(S)$ when $f: 2^M \rightarrow \RR$ with $f(\emptyset) = 0$ is a monotone submodular function.
\end{theorem}

\subsection{Proof of Lemma~\ref{lem:submod}}\label{subsec:submod}

We prove the following more general statement which subsumes Theorem~\ref{thm:submod} as a special case. Let an {\em $\alpha$-approximate oracle} in Line 3 of Algorithm~\ref{alg:submodular} be an algorithm that returns any $i \in M \setminus S$ such that $f(S \cup \{ i \}) - f(S) \geq \frac{1}{\alpha} \cdot \max_{i \in M \setminus S} \{ f(S \cup \{ i \}) - f(S) \}$ : 


\begin{lemma}\label{lem:submodularapprox}
Algorithm~\ref{alg:submodular} with an $\alpha$-approximate oracle in Line~\ref{alg:submodular:ln:oracle} is a $(1 - e^{- 1 / \alpha})$-approximation algorithm for any monotone submodular function maximization problem~\eqref{prob:submod}.
\end{lemma}
\begin{proof}
Let $i_j$ be the element selected in the $j$-th iteration and $S_j = \{ i_1, \ldots, i_j \}$, with $S_0 = \emptyset$. Let $O = \{ o_1, \ldots, o_k \}$ be the optimal solution.

\begin{claim}
For $0 \leq j \leq k-1$, $f(S_{j+1}) - f(S_j) \geq \frac{1}{\alpha k} (f(O) - f(S_j))$.
\end{claim}
\begin{proof}
Note that
\begin{align*}
f_{S_j}(O) & \leq \sum_{o \in O \setminus S_j} f_{S_j}(o)  \\
	& \leq | O \setminus S_j | \cdot \max_{o \in O \setminus S_j} f_{S_j}(o) \\
	& \leq k \cdot \max_{o \in O \setminus S_j} f_{S_j}(o) \\
	& \leq \alpha k \cdot f_{S_j}(i_{j+1}) \\
	& = \alpha k \cdot (f(S_{j+1}) - f(S_j)),
\end{align*}
where the first inequality follows from submodularity and the last one from the $\alpha$-approximate oracle.
Further, we note that $f_{S_j}(O) = f(S_j \cup O) - f(S_j) \geq f(O) - f(S_j)$ by monotonicity. Putting the two inequalities together, we get the claim.
\end{proof}

\begin{claim}
For $1 \leq j \leq k$, $f(O) - f(S_j) \leq \left( 1 - \frac{1}{\alpha k} \right)^j f(O)$.
\end{claim}
\begin{proof}
We prove this using induction. The base case is straightforward to check. For $j > 1$, note that
\begin{align*}
f(O) - f(S_j) & = f(O) - f(S_{j-1}) - [ f(S_j) - f(S_{j-1}) ] \\
	& \leq [ f(O) - f(S_{j-1}) ] \left( 1 - \frac{1}{\alpha k} \right).
\end{align*}
\end{proof}

From the claims, it follows that $f(O) - f(S_k) \leq ( 1 - \frac{1}{\alpha k} )^k f(O) \leq e^{- 1 / \alpha} f(O)$. Equivalently, $( 1 - e^{- 1 / \alpha} ) f(O) \leq f(S_k)$.
\end{proof}

\section{Proof of Lemma~\ref{lem:greedy}}\label{app:greedy}

Note that in the special cases of the unit-interval and unit-disk graphs (i.e., $\gamma = 1$), we get constant approximations, $3$ and $9$, respectively. We can further improve the approximation ratios to $2$ and $5$ by a more direct packing argument (e.g., \cite{MBH+95} for unit-disk graphs).

\begin{proof}
(Interval graphs)
For simplicity, we assume vertices (i.e., the intervals) do not coincide; if two do coincide, then we only keep the one with greater weight. 
Let $I_\alg$ and $I_\opt$ be the independent sets returned by $\alg$ and the optimal algorithm. 
It is sufficient to consider each $v \in I_\opt$ that is not in $I_\alg$. 
Note that $v$ intersects some vertex in $I_\alg$, otherwise it would be added to the solution by the greedy algorithm. 
Then, there exists $v' \in I_\alg$ that intersects $v$ such that $w(v') \geq w(v)$.
Finally, note that each $v' \in I_\alg$ can by intersected by at most $2 + \gamma$ vertices in $I_\opt$. 
It is sufficient to show that the interval graphs are $(2 + \gamma)$-claw free graphs. In the worst case, an interval of length $l_\max$ intersects mutually disjoint intervals of length $l_\min$. Since the intersecting intervals are packed inside a horizontal range of length $2 l_\min + l_\max$, there can be at most $\frac{2 l_\min + l_\max}{l_\min} = 2 + \gamma$ of them. 

(Disk graphs)
Essentially, the same argument works by noting that the disk graphs are $(2 + \gamma^2)$-claw free. In the worst case, a disk of radius $l_\min$ intersects mutually disjoint disks of radius $l_\min$. The smaller disks are contained in a circular range of radii $2 l_\min + l_\max$ around the center of the bigger disk. Hence, there are at most $\frac{\pi \cdot (2 l_\min + l_\max)^2}{\pi \cdot {l_\min}^2} = (2 + \gamma)^2$ smaller disks. 

(Bounded degree-$d$ graphs)
A bounded degree-$d$ graph is $(d+1)$-claw free by definition, and the argument works directly.
\end{proof}

\section{LP Relaxations for the Network Routing Problem}\label{app:network}

The following is the primal-dual linear programming relaxations for the network unicast/multicast routing problem, stated in \pref{sec:network}:

\begin{align*}
\textstyle \max_x ~~ & \textstyle \sum_i v(i) \cdot \left( \sum_{S \in \mathcal{S}_i} x(S) \right) \\ 
& \textstyle \sum_{S : S \in \mathcal{S}, e \in S} q_S(e) x(S) \leq c(e), \quad \forall e \in E \\
& \textstyle \sum_{S \in \mathcal{S}_i} x(S) \leq 1, \quad \forall i \\
& x \geq 0 & & \\
\\
\textstyle \min_{y,z} ~~ & \textstyle \sum_e c(e) y(e) + \sum_i z(i) \\
& \textstyle z(i) + \sum_{e \in S} q_S(e) y(e) \geq v(i), \quad \forall i, S \in \mathcal{S}_i \\
& z, y \geq 0 \\
\end{align*}

\section{LP Relaxations for the Set Cover Problem}\label{app:setcover}

The following is the primal-dual linear programming relaxations for the set cover problem, stated in \pref{sec:setcover}:

\begin{align*}
\textstyle \min_x ~~ & \textstyle \sum_i c(i) x(i) & 
\textstyle \max_y ~~ & \textstyle \sum_e y(e) \\
& \textstyle \sum_{i : e \in S_i} x(i) \geq 1, \quad \forall e \in E & 
& \textstyle \sum_{e : e \in S_i} y(e) \leq c(i), \quad \forall i \in N \\
& x \geq 0 & & y \geq 0
\end{align*}

\end{document}